%% file: ms.tex
\documentclass[runningheads]{llncs}

\usepackage{times}
\usepackage{amsfonts,amstext,amsmath,amssymb}
\usepackage{stmaryrd}
\usepackage{xcolor}
\usepackage{epsfig}
\usepackage{picinpar}
\usepackage{bm}

\usepackage{oz}
\usepackage{traces}
\newcommand{\Nil}{TMPNIL}% replaced in rgref
\usepackage{rgref}

\newcommand{\draftonly}[1]{}

\newcommand{\anegate}{\cnegate}
\newcommand{\tnegate}{\mathop{\neg}}
\newcommand{\atomid}{{\cal E}}
\newcommand{\AtomicSteps}{\mathcal{A}}
\mathchardef\mhyphen="2D
\newcommand{\pguard}[1]{(\mathop{\Keyword{\pi \mhyphen restrict}} #1)}
\newcommand{\eassume}[1]{(\mathop{\Keyword{\epsilon \mhyphen assm}} #1)}

\newcommand{\quint}[5]{\{#1,#2\}#5\{#3,#4\}}
\newcommand{\quintprgqc}{\quint{p}{r}{g}{q}{c}}

\renewenvironment{proof}{\par\noindent\textbf{Proof.}}{\hfill$\Box$}
\newcommand{\Pevent}[1]{\pi(#1)}
\newcommand{\Eevent}[1]{\epsilon(#1)}
\renewcommand{\interleave}{\mathop{|\hspace*{-0.2ex}|\hspace{-0.2ex}|}}

\newcommand{\refeqn}[1]{(\ref{eqn-#1})}

\newcommand{\reflemma}[1]{Lemma~\ref{L-#1}}

\newcommand{\synchpstep}[1]{\Pevent{#1}}
\newcommand{\synchestep}[1]{\Eevent{#1}}
\newcommand{\synchpstepA}{\synchpstep{A}}
\newcommand{\synchestepA}{\synchestep{A}}
\newcommand{\synchpstepa}{\synchpstep{a}}
\newcommand{\synchestepa}{\synchestep{a}}

\newcommand{\synchestepE}{\atomid}

\newcommand{\spstep}[1]{\synchpstep{#1}}
\newcommand{\sestep}[1]{\synchestep{#1}}
\newcommand{\spstepA}{\synchpstepA}
\newcommand{\sestepA}{\synchestepA}
\newcommand{\spstepa}{\synchpstepa}
\newcommand{\spstepac}{\synchpstep{\ecompl{a}}}
\newcommand{\sestepa}{\synchestepa}

\newcommand{\sestepE}{\atomid}

\newcommand{\ssilent}{\spstep{\silent}}

\newcommand{\pl}{\parallel}

\newcommand{\plE}[1]{\underset{#1}{\pl}}
\newcommand{\plA}{\plE{A}}

\newcommand{\Event}{Event}

\newcommand{\silent}{\iota}

\newcommand{\ecompl}[1]{\overline{#1}}

\newcommand{\ecompla}{\ecompl{a}}

\newcommand{\atev}[1]{\atomic{#1}}
\newcommand{\ateva}{\atev{a}}
\newcommand{\atevac}{\atev{\ecompla}}
\newcommand{\atevs}{\atev{\silent}}

\newcommand{\Hide}[2]{{#2}/_{#1}}
\newcommand{\Res}[2]{{#2\backslash#1}}

\newcommand{\AProc}[2]{#1{:}#2}
\newcommand{\AProcA}[1]{\AProc{A}{#1}}
\newcommand{\AProcAp}{\AProcA{p}}

\newcommand{\mult}{\cross}
\newcommand{\One}{\mathbf{1}}
\newcommand{\Inverse}[1]{{#1}^{-1}}

\definecolor{Red}{rgb}{1.,0.,0.}
\definecolor{Blue}{rgb}{0.,0.,1.}
\definecolor{Pink}{rgb}{1.,0.75,0.8}
\definecolor{Green}{rgb}{0.2.,0.5,0.2}

\begin{document}
\title{An algebra of synchronous atomic steps%: 
\thanks{This work is supported by Australian Research Council (ARC) Discovery Project DP130102901
and the UK EPSRC ‘Taming Concurrency’ research grant.}}
\author{Ian J. Hayes\inst{1} \and 
Robert J. Colvin\inst{1} \and 
Larissa A. Meinicke\inst{1} \and 
Kirsten Winter\inst{1} \and \\
Andrius Velykis\inst{2}%
}
\authorrunning{I. J.~Hayes et al.  \draftonly{(\today)}}
\institute{School of 
Information Technology and Electrical Engineering, \\ 
The University of Queensland, Australia
  \and School of Computing Science, Newcastle University, UK 
  \draftonly{\vspace*{2ex}}}

\maketitle

\begin{abstract}
\sloppy
This research started with 
an algebra for reasoning about rely/guarantee concurrency for a shared
memory model. The approach taken led to a more \emph{abstract algebra of
atomic steps}, in which atomic steps synchronise (rather than
interleave) when composed in parallel. The algebra of rely/guarantee
concurrency then becomes an interpretation of the more abstract
algebra. Many of the core properties needed for rely/guarantee
reasoning can be shown to hold in the abstract algebra where their
proofs are simpler and hence allow a higher degree of
automation. Moreover, the realisation that the synchronisation
mechanisms of standard process algebras, such as CSP and CCS/SCCS, can
be interpreted in our abstract algebra gives evidence of its unifying
power. The algebra has been encoded in Isabelle/HOL to provide a basis
for tool support.

\end{abstract}

\input{intro}

\input{generalAlgebra}

\input{tests}

\input{atomicSteps}

\input{guarantees}

\input{processAlgebrasUsingAtomicSteps}

\input{conclusion}

\paragraph{Acknowledgements.}
This work has benefited from input from Cliff Jones and Kim Solin.

\bibliographystyle{alpha}
\bibliography{ms}

\input{proofs}

\end{document}

%% file: intro.tex
\section{Introduction}

Our goal is to provide better methods for deriving concurrent programs from abstract 
specifications, and to provide tool support for compositional reasoning about their correctness.  
The rely/guarantee approach of Jones \cite{jones81d,jones83a} achieves compositionality by abstracting the 
interference a process can tolerate from and inflict on its environment.  
A \emph{rely} condition $r$ is a binary relation between states 
that represents an assumption bounding the interference that a process $p$ can tolerate from its environment. 
If the environment fails to meet its obligation $r$,
$p$ may deviate from its specification and show erratic behaviour (i.e.\ abort). 
A \emph{guarantee} condition $g$ is the corresponding notion 
that bounds the interference inflicted on its environment by $p$. 
For a system of parallel processes to function correctly, 
each process's guarantee must imply the rely of every other parallel process. 
These concepts can be captured uniformly 
(and hence the manipulation of process terms kept simple)
in a framework in which both the steps of a process and the
steps of its environment are explicitly represented.

The semantic model for rely/guarantee reasoning suggested by Aczel is one such framework \cite{Aczel83,DeRoever01}.
In this model, parallel composition \emph{synchronises} a program step of one process
with an environment step of another, to give a program step of their composition.
Aczel's approach, of insisting each step of one process is
synchronised with a step of the other process, differs from the
commonly used approach of interleaving atomic steps of processes
(except when they communicate), e.g.\ CCS \cite{CaC}, CSP
\cite{Hoare85} and ACP \cite{BergstraKlop84,BergstraKlop85}.  
Aczel's approach is closer to Milner's Synchronous CCS (SCCS) \cite[Section 9.3]{CaC}
and Meije (the calculus at the basis of the synchronous programming language Esterelle)
\cite{BerryCosserat85}.

Our methodology is to develop a refinement calculus for concurrent programs
that lifts rely and guarantee conditions to 
commands\footnote{We use the terms \emph{command}, \emph{program} and \emph{process} synonymously.}
\cite{FACJexSEFM-14,HayesJonesColvin14TR}
(from parameters to the notion of correctness).
That allows algebraic reasoning about concurrent programs
in a rely/guarantee style.
To this end
we have designed a \emph{Concurrent Refinement Algebra} (CRA)
to support the rely/guarantee approach \cite{AFfGRGRACP}.
In exploring the laws in CRA, 
we discovered that atomic steps have 
specific
 algebraic properties 
that can be captured in an \emph{abstract algebra of atomic steps} 
which is embedded in CRA.

The abstract algebra of atomic steps delivers  a range of useful properties for manipulating process terms. 
For example, based on the notion of atomic steps the parallel composition of processes 
can be simplified as follows
\begin{eqnarray}
  (a \Seq c) \parallel (b \Seq d) & = & (a \parallel b)  \Seq  (c \parallel d)~, \label{atomic-interchange}
\end{eqnarray}
where $a$ and $b$ are atomic steps and $c$ and $d$ are arbitrary
processes.  Note that the above equivalence does not hold if $a$ and
$b$ are arbitrary processes.  
For an interleaving operator $\interleave$ the corresponding law is the more complicated:
\begin{eqnarray}
  (a \Seq c) \interleave (b \Seq d) & = & a  \Seq (c \interleave b \Seq d) \nondet b \Seq (a \Seq c \interleave d)~.
\end{eqnarray}
In (\ref{atomic-interchange}), parallel composition of two atomic steps $a$ and $b$ gives an atomic step $a \parallel b$,
where the interpretation of $a \parallel b$ depends on the particular model.
As a consequence, the algebra can be applied to a range of models. For example, as well
as allowing an Aczel-trace  model to support shared variable
concurrency, communication in process algebras such as CSP and CCS/SCCS can be
interpreted in the abstract algebra and hence it provides a foundation for a
range of concurrency models.

\emph{Kleene Algebra with Tests} (KAT) 
by Kozen~\cite{kozen97kleene}
combines Kleene algebra (the algebra of regular expressions \cite{Conway71})
with a Boolean sub-algebra representing tests.
KAT supports sequential programs with conditionals and finite iterations (partial correctness).
The \emph{Demonic Refinement Algebra} (DRA) of von Wright \cite{Wright04}
generalises Kozen's work to support possibly infinite iteration
and with that the concept of aborting behaviour.
The approach used in this paper is based on that of von Wright 
in order to faithfully capture Jones' theory, in particular his rely condition.

\emph{Concurrent Kleene Algebra} (CKA) \cite{DBLP:journals/jlp/HoareMSW11} 
adds a parallel operator to Kleene algebra to support sequential and parallel programs.
Prisacariu's \emph{Synchronous Kleene Algebra} (SKA) \cite{Pris10} extends Kleene algebra
with a synchronous parallel operator similar to that in Milner's SCCS \cite{CaC}.
Like Milner he proposes a specific interpretation of the parallel composition of atomic steps.
In contrast to both CKA and SKA, our \emph{Concurrent Refinement Algebra} \cite{AFfGRGRACP}, 
which we use as a basis for this work,
adds a parallel operator to the sequential algebra DRA 
(rather than Kleene algebra).

The major contribution of this paper is an \emph{algebra of atomic steps}
which introduces a synchronous parallel operator for atomic steps.
The interpretation of two atomic steps acting in parallel, however, is left open, 
hence allowing a range of different models (including those of Milner and Prisacariu).
Further, atomic steps are treated as a Boolean sub-algebra
(similar to the way in which tests are treated as a Boolean sub-algebra in KAT).
Hence the Concurrent Refinement Algebra (CRA)  contains both
a sub-algebra of tests and a sub-algebra of atomic steps 
(as illustrated in Figure~\ref{lattices} via their lattices).
Separating out these sub-algebras enables one to prove properties that are specific
to atomic steps using the full power of a Boolean algebra. 
This raises the level of support for reasoning about programs provided by our algebra, 
as well as the level of automation that is possible for the mechanised proof support 
by the theorem prover Isabelle.

\begin{figure}[t]
 \centering
   \includegraphics[viewport=31 606 332 828,width=0.48\linewidth]{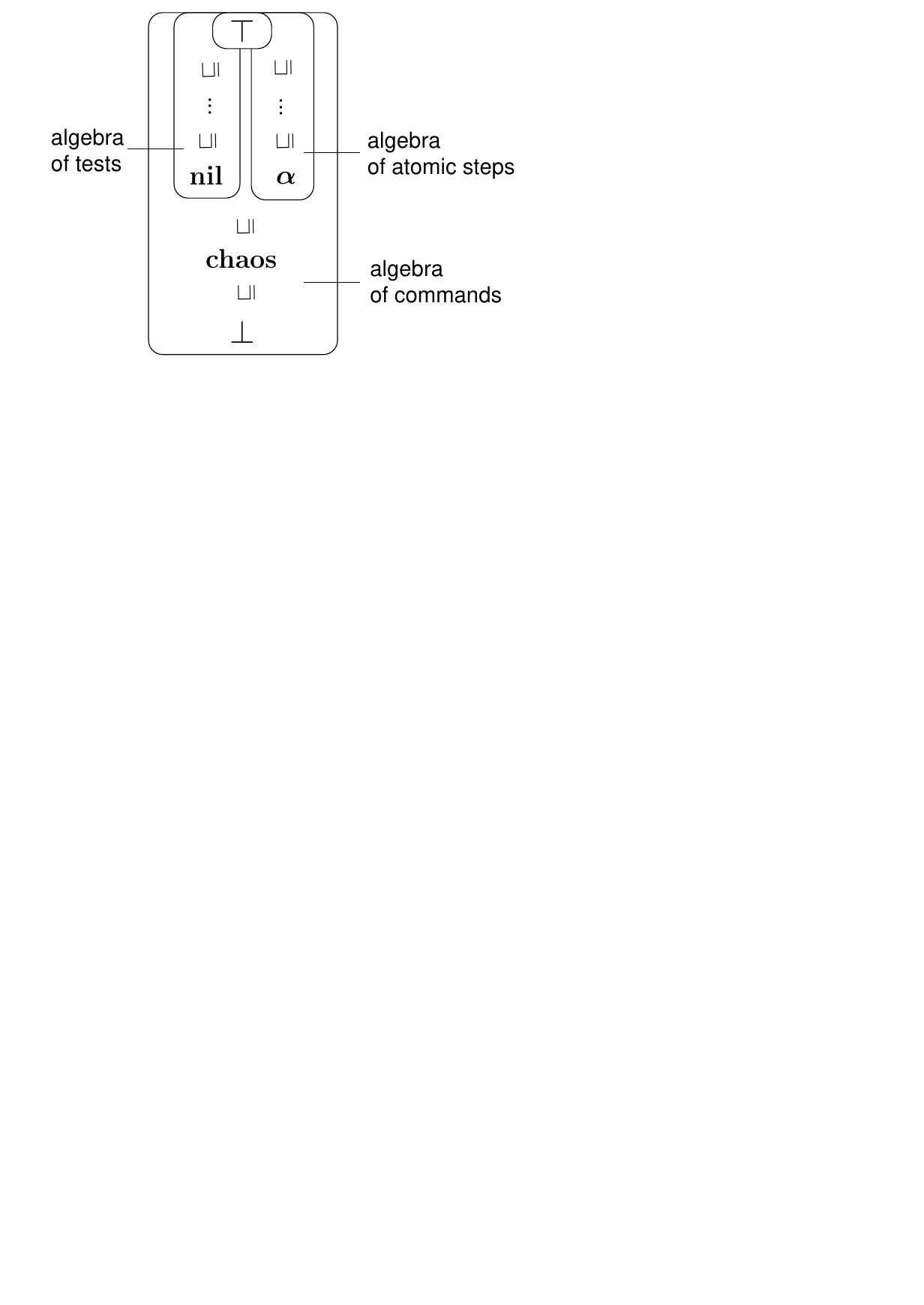}
  \caption{The Concurrent Refinement Algebra and its sub-algebras}
  \label{lattices}
\end{figure}
To build the algebra, 
we start in Section~\ref{S-general-algebra} with CRA for reasoning about commands in general.
Commands include a sub-lattice of tests (Section~\ref{S-tests})
as well as 
a second sub-lattice of atomic steps (Section~\ref{S-abstract-atomic}), 
the novel contribution of this paper.  
Section~\ref{S-relational-atomic} gives an interpretation of the abstract algebra
based on Aczel's trace model. 
A simplified treatment of relies and guarantees 
is outlined in Section~\ref{S-rely-guarantee}.
Section~\ref{S-process-algebras} illustrates how the communication models of CCS, CSP 
and SCCS
can be interpreted in our abstract algebra of atomic steps.

%% file: generalAlgebra.tex
\section{Concurrent Refinement Algebra}\label{S-general-algebra}

A Concurrent Refinement Algebra  (CRA) is defined as the following structure
\begin{eqnarray*}
(\cal{C}, \nondet, \sqcup, ~\Seq~, \parallel, \bot, \top, \Nil, \Skip)
\end{eqnarray*}
where the carrier set $\cal{C}$ is the set of \emph{commands}. 
Sequential composition ($\Seq$) has higher precedence than parallel ($\parallel$), 
which has higher precedence than $\nondet$ and $\sqcup$, which have equal precedence.

Commands form a complete distributive lattice 
$(\cal{C}, \nondet, \sqcup, \bot, \top)$ 
with \emph{nondeterministic choice} as the lattice meet ($c\nondet d$),
and \emph{conjunction} of commands as the lattice join ($c \sqcup d$).  
The top of the lattice $\top$ is the infeasible command 
(called ``magic'' in the refinement calculus) 
and the bottom of the lattice $\bot$ is the command that aborts.  
The partial order defined on commands is the refinement relation $c \refsto d$
meaning $c$ is refined (or implemented by) $d$. 
For any commands $c,d \in \cal{C}$, 
$c \refsto d \sdefs (c \nondet d) = c$,
and hence $\bot \refsto c \refsto \top$. 
We refer to this as the \emph{refinement lattice} (see Figure~\ref{lattices}).
Note that since CRA is a \emph{refinement} algebra it uses $\refsto$ as partial order 
instead of Kozen's $\geq$ and hence our lattice of commands is the dual of Kozen's lattice 
(i.e., $\sqcap$ in CRA matches $\sqcup$ in KAT, and $\sqcup$ in CRA matches $\sqcap$ in KAT).
Given commands form a complete lattice, for any monotone function
least/greatest fixed points are well defined.  In particular, fixed
points are used to define iteration operators below.

\newcommand{\Angelic}{\mathop{\textstyle\bigsqcup}}

Sequential composition of commands ($c \Seq d$) is associative and has
identity $\Nil$.  
As an abbreviation, the sequential composition operator may be elided.
Sequential composition has both $\top$ and $\bot$ as left (but not right) annihilators%
\footnote{Here our approach based on DRA differs from approaches based on Kleene algebra, like CKA and SKA, 
in which $\top$ is also a right annihilator.}, 
i.e.\ $\top \SSeq c = \top$ and $\bot \SSeq c = \bot$.
It 
distributes over arbitrary choices on the right (\ref{L-seq-distr-right}), 
\begin{eqnarray}
  (\textstyle\Nondet C) \SSeq d & = & \textstyle\Nondet _{c \in C} (c \SSeq d)~. \label{L-seq-distr-right}
\end{eqnarray}
The iteration of a command is inductively defined as 
$c^0 = \Nil$ and $c^{i+1} = c \SSeq c^i$.
More general iteration operators are captured via greatest ($\nu$) and least
($\mu$) fixed points: $\Fin{c} \sdefs \nu x .\Nil \nondet c \SSeq x$\; for
finite iteration zero or more times, and $\Om{c} \sdefs \mu x . \Nil \nondet c \SSeq x$~ for
finite or 
possibly 
infinite iteration. 
Infinite iteration is defined as $\Inf{c} = \Om{c} \top$. 
The unfolding laws (\ref{L-omega-unfold}) and (\ref{L-finite-unfold})
result from
the fixed point definitions for iterations, 
and
(\ref{L-infinite-unfold}) follows from (\ref{L-omega-unfold}) and the definition of $\Inf{c}$
which also justifies (\ref{L-infinite-annihilates}).
Law~(\ref{L-infinite-unfold-power}) follows from (\ref{L-infinite-unfold}) by induction.
\begin{minipage}{0.49\textwidth}
\begin{eqnarray}
  \Om{c} &=& \Nil \nondet c \SSeq \Om{c} \label{L-omega-unfold} \\
  \Fin{c} &=& \Nil \nondet c \SSeq \Fin{c} \label{L-finite-unfold} \\
  \Inf{c} &=& c \SSeq \Inf{c} \label{L-infinite-unfold}
\end{eqnarray}
\end{minipage}
\begin{minipage}{0.49\textwidth}
\begin{eqnarray}
  \Inf{c} &=& c^i \SSeq \Inf{c} \label{L-infinite-unfold-power} \\
  \Inf{c}\SSeq d &=& \Inf{c} \label{L-infinite-annihilates} 
\end{eqnarray}
\end{minipage}
\\

Some models also distribute sequential composition 
over non-empty choices on the left (\ref{L-seq-distr-left}) 
(i.e., in refinement calculus terms the operator is conjunctive).
\begin{eqnarray}
D \neq \{\} \implies c \SSeq (\textstyle\Nondet D) & = & \textstyle\Nondet_{d \in D} (c \SSeq d) \label{L-seq-distr-left}
\end{eqnarray}

This axiom is not assumed to generally hold in CCS and CSP but
it holds for our relational model in Section~\ref{S-relational-atomic}
and is required to show laws~(\ref{L-isolation}) and (\ref{L-finite-iteration}).
Laws~(\ref{L-iteration1}) and (\ref{L-iteration2}) 
follow from (\ref{L-isolation}), (\ref{L-infinite-annihilates}) and (\ref{L-infinite-unfold}).
\\
\begin{minipage}{0.5\textwidth}
\begin{eqnarray}
  \Om{c} &=& \Fin{c} \nondet \Inf{c} \label{L-isolation} \\
  \Fin{c} &=& \textstyle\Nondet_{i \in \nat} c^i\label{L-finite-iteration}
\end{eqnarray}
\end{minipage}
\begin{minipage}{0.5\textwidth}
\begin{eqnarray}
  \Om{c} \SSeq d &=& \Fin{c} \SSeq d \nondet \Inf{c} \label{L-iteration1} \\
  c \SSeq \Om{c} \SSeq d &=& c\SSeq \Fin{c} \SSeq d \nondet \Inf{c} \label{L-iteration2}
\end{eqnarray}
\end{minipage}
\\

Parallel composition of commands is associative,
commutative, has the identity $\Skip$, and $\top$ serves as an annihilator: 
$c \parallel \top = \top$.
Parallel distributes over non-deterministic choice of any set of commands,
$c \parallel (\Nondet D) =\Nondet_{d \in D} (c \parallel d)$.
Note the identities for sequential and parallel composition,
$\Nil$ and $\Skip$ respectively, differ.
However, they are related by $\Skip \refsto \Nil$
and $\Nil \parallel \Nil = \Nil$.

%% file: tests.tex
\section{The Boolean sub-algebra of tests}
\label{S-tests}

Tests are special commands that are used to model conditionals and loops and hence form an
essential construct when reasoning about programs.
Assume $t$ is a test, $\tnegate{t}$ is its negation, and $c$ and $d$ are
commands, 
an abstract algebraic representation of conditionals and
while loops for sequential programs is given by
\[
  \text{\bf if }  t ~\text{\bf then } c ~\text{\bf else } d \sdefs t \SSeq c \nondet \tnegate t \SSeq d
  ~~~\text{ and }~~~
  \text{\bf while } t ~\text{\bf do } c \sdefs \Om{(t \SSeq c)} \SSeq \tnegate t
\]
Blikle \cite{Blikle78} used this style of representation of programs in a relational algebra 
and \cite{gardinerfacj93} and \cite{Wright04} in the refinement calculus.
Kozen \cite{kozen97kleene} provided a more abstract 
\emph{Kleene Algebra with Tests} (KAT) as a framework for reasoning about
programs.  Kleene algebra is the algebra of regular expressions, 
where for the interpretation as programs,
alternation becomes non-deterministic choice with unit $\top$,
concatenation becomes sequential composition with unit $\Nil$, and
iteration becomes finite iteration of commands.
Tests, in Kozen's approach, form a Boolean sub-algebra within the Kleene algebra.  

We follow this construction for the Concurrent Refinement Algebra.
That means in CRA tests form a subset of commands for which a
negation operator $\tnegate$ is defined. This results in an extended algebra
\begin{eqnarray*}
(\cal{C}, \cal{B}, \nondet, \sqcup, ~\Seq~, \parallel, \bot, \top, \Nil, \Skip, ~\tnegate)
\end{eqnarray*}
where the additional carrier set $\cal{B}$ is the set of \emph{test commands} ($\cal{B} \subseteq \cal{C}$). 
As in Kozen's work, tests form a Boolean algebra
$
(\cal{B}, \nondet, \sqcup, ~\tnegate~, \top, \Nil)
$
which is a \emph{sub-lattice} of commands (see Figure~\ref{lattices}).

The sub-lattice of tests shares its top element (the false test) with the top of
the lattice of commands, $\top$, but does not share its bottom element,
the true test, that instead corresponds to the command $\Nil$, 
that has no effect and immediately terminates.
Tests are closed under lattice meet and join,
as well as sequential and parallel composition as both are defined via the join operator $\sqcup$ on commands. 
For any $t$ and $t'$ in $\cal{B}$,
\vspace*{-2ex}\\
\begin{minipage}{0.5\textwidth}
\begin{eqnarray}
  t \SSeq t' &=& t \sqcup t'      \label{test-seq-test}
\end{eqnarray}
\end{minipage}
\begin{minipage}{0.49\textwidth}
\begin{eqnarray}
 t \parallel t' &=&  t \sqcup t' \label{test-par-test}
\end{eqnarray}
\end{minipage}
\vspace*{0.5ex}\\
where the join of two test acts as logical conjunction. 
Property (\ref{test-par-test}) can be generalised to 
the following interchange axiom.
For any commands $c$ and $d$ in $\cal{C}$ and any tests $t$ and $t'$ in $\cal{B}$
the following hold.
\vspace*{-2ex}\\
\begin{minipage}{0.5\textwidth}
\begin{eqnarray}
(t \SSeq c) \parallel (t' \SSeq d)  = (t \sqcup t') \SSeq (c \parallel d)
\end{eqnarray}
\end{minipage}
\begin{minipage}{0.49\textwidth}
\begin{eqnarray}
(t \SSeq c) \sqcup (t' \SSeq d)  = (t \sqcup t') \SSeq (c \sqcup d)
\end{eqnarray}
\end{minipage}
\vspace*{1ex}

A range of useful laws follow from this axiomatisation that help
simplifying program terms involving tests. 

Tests also give rise to the concept of \emph{assertions}
(preconditions) \cite{Wright04,Solin07}. 
The assertion corresponding to a test $t$ is a command which
terminates if the test holds and aborts if the test does not hold,
i.e., $\Assert{t} = t \nondet \tnegate{t} \SSeq \bot$.

%% file: atomicSteps.tex
\section{Abstract atomic steps}\label{S-abstract-atomic}

This section gives an abstract algebra for the subset of commands 
that correspond to atomic steps.
This algebra delivers core properties of atomic steps 
(that do not hold for commands in general)
under only a few assumptions about the form of atomic steps.
Atomic steps are closed under parallel composition but
the parallel composition of atomic steps, $a \parallel b$, is left uninterpreted.
Lifting these properties to the level of an abstract algebra results in simpler proofs and allows for their reuse in different interpretations.
Section~\ref{S-relational-atomic} forms an interpretation of the atomic step algebra 
that corresponds to Aczel's program and environment steps and defines parallel composition of atomic steps 
in detail.
Section~\ref{S-process-algebras} on the other hand, uses the atomic step algebra to capture CCS-style as well as CSP-style communication of events, which resides in a very different domain.

In the same manner that tests form a sub-lattice of commands, 
the set of atomic steps, $\AtomicSteps$, forms a sub-lattice of commands which is a Boolean algebra
and shares the lattice meet and join of commands (see Figure~\ref{lattices}).
The top of the sub-lattice is the same as the top of the command lattice ($\top$)
but the bottom of the sub-lattice is the new command $\cstepd$,
that can be thought of as the non-deterministic choice between all possible atomic steps.
In fact, tests and atomic steps  share only one element ($\top$) and hence
\begin{eqnarray}
 \cstepd \sqcup \Nil = \top \label{A-test-atomic-sup}~.
\end{eqnarray}

The term \emph{step} is used exclusively for an atomic step.
Steps are closed under lattice meet and join as well as parallel composition
(but not sequential composition).
As for commands, the meet corresponds to non-deterministic choice, $a \nondet b$,
and can behave as either $a$ or $b$.
The join of two steps, $a \sqcup b$, can be thought of as a step 
that both $a$ and $b$ agree to do.
(In Section~\ref{S-relational-atomic} this
corresponds to the intersection of the sets of primitive steps $a$ and $b$ can make.)

Because $\AtomicSteps$ forms a Boolean algebra,
all of the laws of Boolean algebra are available to manipulate combinations of steps
not involving sequential composition.
The theorem prover Isabelle directly supports forming such an interpretation
and hence the theory of Boolean algebra can be re-used for $\AtomicSteps$.
This is a significant saving, as the laws of Boolean algebra do not need to be reproven.

In addition, atomic steps are assumed to have an \emph{identity}, $\atomid$, of parallel composition,
giving the following axiom.
\vspace*{-1ex}
\begin{eqnarray}
  a \parallel \atomid & = & a  \label{A-atomic-parallel-identity}
\end{eqnarray}
Prefixing a command $c$ with $\atomid$, i.e.\ $\atomid \SSeq c$, allows the process to
wait one step before behaving as $c$,
and $\Om{\atomid} \SSeq c$ allows it to wait any number of steps (including 0).
The step $\atomid$ can be interpreted as a placeholder for one step taken by its environment.

Besides laws for reasoning about atomic steps in isolation,
one needs laws that allow reasoning about their interaction with non-atomic commands.
A small set of additional axioms is used as the basis of these laws.
The approach taken to handling parallel composition is not the usual interleaving of steps, 
rather each step of one process must synchronise with a step of the other process.
If $a$ and $b$ cannot synchronise then $a \parallel b$ is infeasible ($\top$).
For steps $a$ and $b$, and any commands $c$ and $d$,
we assume the following axioms.
\vspace*{-4ex}\\
\begin{minipage}[t]{0.49\textwidth}
\begin{eqnarray}
  a \SSeq c \parallel b \SSeq d & = & (a \parallel b) \SSeq (c \parallel d) \label{A-parallel-interchange-sequential}\\
  a \SSeq c \sqcup b \SSeq d & = & (a \sqcup b) \SSeq (c \sqcup d) \label{A-join-interchange-sequential}
\end{eqnarray}
\end{minipage}
\begin{minipage}[t]{0.49\textwidth}
\begin{eqnarray}
  a \SSeq c \parallel \Nil & = & \top \label{A-atomic-parallel-nil}\\
  a \SSeq c \sqcup \Nil & = & \top \label{A-atomic-join-nil}
\end{eqnarray}
\end{minipage}\\

The interchange axioms (\ref{A-parallel-interchange-sequential}) and (\ref{A-join-interchange-sequential}) 
become refinements from left to right if $a$ and $b$ are allowed to be arbitrary commands
(which corresponds to the weak interchange law in CKA \cite{DBLP:journals/jlp/HoareMSW11}).
The abstract algebra does not define the details of parallel composition of pairs of steps.
(See the relational interpretation of the algebra in Section~\ref{S-relational-atomic}
for one example of defining parallel composition of atomic steps.)
The command, $\Nil$, that terminates immediately without making any steps whatsoever 
cannot synchronise with a process that makes at least one step,
i.e.\ (\ref{A-atomic-parallel-nil}) and (\ref{A-atomic-join-nil}).

The \emph{negation} operator ($\anegate$) for atomic steps satisfies the following axioms
of a Boolean algebra.
Steps $a$ and $\anegate{a}$ have no common behaviour (\ref{anegate-disjoint})
and $\anegate{a}$ has all the step behaviours that $a$ does not have (\ref{anegate-complement}).
\vspace*{-2.5ex}
\\
\begin{minipage}{0.49\textwidth}
\begin{eqnarray}
  a \sqcup \anegate a & = & \top \label{anegate-disjoint}
\end{eqnarray}
\end{minipage}
\begin{minipage}{0.49\textwidth}
\begin{eqnarray}
  a \nondet \anegate a & = & \cstepd \label{anegate-complement}
\end{eqnarray}
\end{minipage}
\vspace*{1ex}
\\
Note that negation for tests ($\tnegate$) differs from negation for atomic steps ($\anegate$)
as we have $\tnegate \top = \Nil$ but $\anegate \top = \cstepd$.
The inclusion of a negation operator on steps
allows one to define an equivalent of an assertion for steps on the abstract level.
For any step $a$ define,
\begin{eqnarray}
\Assume{a} \sdef a \nondet (\anegate a) \SSeq \bot~. \label{def-assume}
\end{eqnarray}
\vspace*{-3.5ex}
\\
The command $\Assume{a}$ behaves as $a$ and terminates, or as $\anegate{a}$ and aborts.
It represents an assumption that step $a$ occurs in the sense that any other step 
allows any behaviour to occur after that step.
It provides the basis for rely conditions 
because they specify assumptions about the environment's behaviour (see Section~\ref{S-rely-guarantee}).

\subsection{Canonical representation of commands}

If the primitive commands of our language are tests, atomic steps and $\bot$, and
all other commands are built from these primitives using the operators of the language,
then initially, a command may either terminate immediately, abort or perform some atomic step.
That leads to the
canonical representation theorem,
in which $c$ can terminate if some test $t$ succeeds, abort if some test $t'$ succeeds,
or performs some step $a_i$ followed by some command $c_i$, for some $i \in I$.
\begin{theorem}[canonical-representation] \label{T-canonical-representation}
Any command $c$ can be expressed in the following form
\vspace*{-3ex}
\begin{eqnarray*}
c = t \nondet t'\bot \nondet \textstyle\Nondet_{i \in I} a_i\SSeq c_i
\end{eqnarray*}
where $t$ and $t'$ are tests, and for any $i$ in some (possibly empty) index set $I$, 
$a_i$ is an atomic step not equal to $\top$, and $c_i$ is a command.
\end{theorem}
The proof is conducted by structural induction over commands.
Note that if $c$ cannot terminate immediately, $t$ is $\top$.
If $c$ cannot abort, $t'$ is $\top$.
If $c$ cannot perform any step, $I = \{\}$.
A similar theorem can be found in \cite{Pris10} for SKA.

Because $\atomid$ is the identity of parallel for a single step,
$\Om{\atomid}$ acts as the identity of any sequence of steps
and hence $\Om{\atomid}$ is the identity of parallel, i.e.\ $\Skip = \Om{\atomid}$.
\begin{lemma}[atomic-identity-iteration]\label{L-atomic-identity-iteration}
\(
  ~~\Om{\atomid} \parallel c = c
\)
\end{lemma}
The proof makes use of Theorem~\ref{T-canonical-representation} to express $c$ in canonical form
(the proof is included in the appendix of \cite{TRatomicSteps}).

\subsection{Properties of iterations of atomic steps}
\label{S-atomic-iteration}

In addition to defining programming statements such as while loops, 
iterators are used to build specifications from atomic steps.  
For instance commands corresponding to Jones' rely and guarantee concepts are 
constructed as iterations of relatively straightforward commands 
that make assumptions about the steps of the environment and
constrain the steps of the program, respectively
(see Section~\ref{S-rely-guarantee}).  Below we provide
some useful properties of iterations of atomic steps.

Because $\Nil$ performs no steps, if it is run in parallel with a (possibly) finite iteration,
the composition cannot perform any steps but can terminate and hence equals $\Nil$.
If $\Nil$ is run in parallel with an infinite iteration, 
the combination cannot perform any steps but cannot terminate,
and hence equals the infeasible command $\top$.
\begin{lemma}[atomic-iteration-nil]\label{L-atomic-iteration-nil}
\begin{eqnarray*}
  \Fin{a} \parallel \Nil = \Nil
\hspace*{1cm}
  \Om{a} \parallel \Nil = \Nil
\hspace*{1cm}
  \Inf{a} \parallel \Nil = \top
\end{eqnarray*}
\end{lemma}

\begin{proof}
The properties follow from axiom (\ref{A-atomic-parallel-nil}) 
using unfolding of the iterations 
(i.e. $\Fin{a} = \Nil \nondet a \SSeq \Fin{a}$, 
$\Om{a} = \Nil \nondet a \SSeq \Om{a}$
and
$\Inf{a} = a \SSeq \Inf{a}$).
\end{proof}

For the following lemmas, 
let $a$ and $b$ be atomic steps, and $c$ and $d$ any commands.
Axiom (\ref{A-parallel-interchange-sequential}) can be extended to iteration $i$ times 
as given in the following lemma, which is proven by induction on $i$.
\begin{lemma}[atomic-iteration-power]\label{L-atomic-iteration-power}
\(
 a^i \SSeq c \parallel b^i \SSeq d  ~=~  (a\parallel b)^i \SSeq (c \parallel d)
\)
\end{lemma}
Choosing $c$ and $d$ to both be $\Nil$ gives the corollary that $a^i \parallel b^i = (a \parallel b)^i$.

For all further lemmas in this sub-section, we assume that sequential composition 
is conjunctive (\ref{L-seq-distr-left}) and 
hence that properties (\ref{L-isolation}) and (\ref{L-finite-iteration}) hold.
Two useful properties are the following.
\\
\begin{minipage}{0.4\textwidth}
\begin{eqnarray}
  \Fin{a} \parallel \Fin{b} = \Fin{(a \parallel b)} \label{atomic-fin-parallel-fin}
\end{eqnarray}
\end{minipage}
\begin{minipage}{0.59\textwidth}
\begin{eqnarray}
  \Inf{a} \parallel \Inf{b} & = & \Inf{(a \parallel b)} \label{A-atomic-iteration-infinite-infinite}
\end{eqnarray}
\end{minipage}\\

Property (\ref{atomic-fin-parallel-fin}) can be proven using the property that
non-deterministic choice over an arbitrary set distributes over parallel. A proof of 
(\ref{A-atomic-iteration-infinite-infinite}) would follow straightforwardly 
if the supremum over an arbitrary set (or even a chain) distributed over parallel,
however, that distribution property does not hold in general. 
We take property (\ref{A-atomic-iteration-infinite-infinite})  as an axiom 
because it does hold in our intended model.
Whether this axiom is independent of the other axioms in our algebra is an open question.

Property (\ref{atomic-fin-parallel-fin}) holds for atomic steps $a$ and $b$
but is only a refinement from left to right if $a$ and $b$ are replaced by arbitrary commands.
Property (\ref{atomic-fin-parallel-fin}) can be generalised to the following lemma
where we take into account that
the number of iterations of $a$ and $b$ might be the same,
or there are more iterations of $a$ than $b$
(and hence the additional iterations of $a$ are in parallel with the start of $d$),
or the symmetric case when there are more occurrences of $b$ than $a$.
\begin{lemma}[atomic-iteration-finite]\label{L-atomic-iteration-finite}
\begin{eqnarray*}
 \Fin{a} \SSeq c \parallel \Fin{b} \SSeq d 
  & = & \Fin{(a \parallel b)} \SSeq ((c \parallel d) \nondet (c \parallel b \SSeq \Fin{b} \SSeq d) \nondet (a \SSeq \Fin{a} \SSeq c \parallel d))
\end{eqnarray*}
\end{lemma}

Isabelle/HOL proofs of these lemmas have been completed.
They may be also found in the appendix of \cite{TRatomicSteps}.
Choosing $c$ and $d$ to both be $\Nil$ gives (\ref{atomic-fin-parallel-fin}) as a corollary.

An infinite iteration in parallel with an initial finite iteration matches the finite iteration
as well as what follows it.
\begin{lemma}[atomic-iteration-finite-infinite]\label{L-atomic-iteration-finite-infinite}
\( \Fin{a} \SSeq c \parallel \Inf{b} = \Fin{(a \parallel b)} \SSeq (c \parallel \Inf{b})
\)
\end{lemma}

Lemma~\ref{L-atomic-iteration-finite} can be extended to initial iterations 
that are either finite or infinite.
\begin{lemma}[atomic-iteration-either]\label{L-atomic-iteration-either}
\begin{eqnarray*}
 \Om{a} \SSeq c \parallel \Om{b} \SSeq d & = &
    \Om{(a \parallel b)} \SSeq ((c \parallel d) \nondet (c \parallel b \SSeq \Om{b} \SSeq d) \nondet (a \SSeq \Om{a} \SSeq c \parallel d))
\end{eqnarray*}
\end{lemma}
Choosing $c$ and $d$ to both be $\Nil$ gives the corollary that $\Om{a} \parallel \Om{b} = \Om{(a \parallel b)}$.

To see the relationship to an interleaving operator,
for any step $a$, define an 
\emph{action}
as $\atomic{a} = \Om{\atomid} \SSeq a \SSeq \Om{\atomid}$, then 
properties of $\atomic{a}$ can be proven using properties of the abstract algebra.
For example, one can derive the following lemma.
\begin{lemma}[atomic-interleaving]\label{L-atomic-interleaving}
\(
  \atomic{a} \parallel \atomic{b} ~=~ \atomic{a \parallel b} \nondet \atomic{a} \SSeq \atomic{b} \nondet \atomic{b} \SSeq \atomic{a}
\)
\end{lemma}
If $a$ and $b$ cannot synchronise (i.e. $a \parallel b = \top$) then
$\atomic{a} \parallel \atomic{b} ~=~ \atomic{a} \SSeq \atomic{b} \nondet \atomic{b}\SSeq \atomic{a}$
which echoes the following property of an interleaving operator:
\(
  a \interleave b  =  a \SSeq b \nondet b \SSeq a .
\)
Hence by including an identity, $\atomid$, for parallel with an atomic step, 
one can represent interleaving properties in the synchronising algebra albeit in a more complex form.
This approach was used by Milner in Synchronous CCS
\cite{Milner83} to allow 
the encoding of 
the better-known process algebra CCS. 
Our identity element takes on a similar role, although
we lift it to
a command 
as opposed to a transition event as in Milner's operational semantics.  
The advantage of the synchronising algebra is that one can represent both 
synchronising events and interleaving events in the one theory.
By using separate program and environment events,
such a theory supports the rely/guarantee approach of Jones for reasoning
about concurrent programs.

\section{Relational atomic steps}\label{S-relational-atomic}

This section examines an interpretation of the abstract atomic step algebra $\AtomicSteps$ in terms
of Aczel's program and environment state transitions.%
\footnote{A semantic model for this interpretation may be found in \cite{DaSMfaWSLwC-TR}.}
The resulting \emph{relational atomic steps} are used to define guarantees and relies
in Section~\ref{S-rely-guarantee}.
This interpretation assumes that sequential composition is conjunctive (\ref{L-seq-distr-left}).

Given a state space $\Sigma$ and a binary relation $r \in \pset(\Sigma \times \Sigma)$,
the command $\cpstepr$ can take an atomic \emph{program} step from state $\sigma$ to $\sigma'$ for any pair of states $(\sigma, \sigma')$ in $r$. 
Similarly, $\cestepr$ is a command that can perform any environment step
from state $\sigma$ to $\sigma'$ whenever $(\ssp) \in r$.
\\
\begin{minipage}{0.49\textwidth}
\begin{eqnarray*}
  \cpstepd & : & \pset(\Sigma \times \Sigma) \fun \AtomicSteps   \label{pstepr}
\end{eqnarray*}
\end{minipage}
\begin{minipage}{0.49\textwidth}
\begin{eqnarray*}
  \cestepd & : & \pset(\Sigma \times \Sigma) \fun \AtomicSteps   \label{estepr}
\end{eqnarray*}
\end{minipage}
\vspace{2ex}
\\
The commands $\cpstep{\emptyset}$ and $\cestep{\emptyset}$ are infeasible, i.e., $\cpstep{\emptyset} = \cestep{\emptyset} = \top$. 
The images of $\cpstepd$ and $\cestepd$ are disjoint except when the relation is empty,
i.e.\ for all $r_1$ and $r_2$,
\begin{eqnarray}
 \cpstep{r_1} \sqcup \cestep{r_2} = \top ~.
\end{eqnarray}
Together $\cpstepd$ and $\cestepd$ form a sub-lattice of commands with two further sub-lattices:
all the $\cpstep{r}$ commands form a sub-lattice and all the $\cestep{r}$ commands form a sub-lattice.

The functions $\cpstepd$ and $\cestepd$ are injective, i.e. different relations map to different commands,
and union of relations maps to a non-deterministic choice between the mappings of the relations
and intersection maps to the supremum in the command ordering.
\vspace*{-2ex}
\\
\begin{minipage}{0.49\textwidth}
\begin{eqnarray}
  r_1 = r_2 & \iff & \cpstep{r_1} = \cpstep{r_2}   \label{pstepr-injective} 
\end{eqnarray}
\end{minipage}
\begin{minipage}{0.49\textwidth}
\begin{eqnarray}
  \cpstep{r_1 \union r_2} & = & \cpstep{r_1} \nondet \cpstep{r_2} \\
  \cpstep{r_1 \int r_2} & = & \cpstep{r_1} \sqcup \cpstep{r_2} 
\end{eqnarray}
\end{minipage}
\vspace*{1ex}
\\
If $r_1 \subseteq r_2$, then
\(
  \cpstep{r_1} \nondet \cpstep{r_2} = \cpstep{r_1 \cup r_2}  =  \cpstep{r_2},
\)
and therefore $\cpstep{r_2} \refsto \cpstep{r_1}$.
Similar laws hold for $\cestepd$ steps.

In this interpretation one can instantiate the test command from Section~\ref{S-tests}
as $\gd{p}$ for $p \in \power \Sigma$,
which succeeds and terminates immediately if $p$ holds but is $\top$ otherwise,
e.g.\ $\gd{\emptyset} = \top$ and $\gd{\Sigma} = \Nil$.
As in the refinement calculus, 
a precondition command $\Pre{p}$ can then be defined as $\Assert{\gd{p}}$,
which equals $\gd{p} \nondet \gd{\lnot p} \SSeq \bot$,
and hence terminates immediately if $p$ holds but aborts otherwise,
e.g.\ $\Pre{\emptyset} = \bot$ and $\Pre{\Sigma} = \Nil$.

%% file: guarantees.tex
\section{Relies and guarantees}\label{S-rely-guarantee}

The rely/guarantee approach of Jones \cite{CoJo07} makes use of a rely condition,
$r$, a binary relation on states that expresses an assumption that every step made 
by the environment of the process satisfies $r$ between its before and after states.
Complementing that, all processes in its environment have a guarantee condition, 
$g$, a binary relation on states that expresses that every program step made by the process satisfies $g$.
For each process, its guarantee condition must imply the rely conditions of all the processes in
its environment.
This section encodes guarantees and relies using the abstract algebra of atomic steps.

\subsection{The guarantee command}\label{S-guarantee}

For a process to ensure a guarantee $g$, 
every atomic program ($\cpstepd$) step made by the program must satisfy $g$.
A guarantee puts no constraints on the environment of the process.
A guarantee command, $\Guar{g}$, is defined in terms of the iteration of
a single step guarantee, $\pguard{g}$, defined as follows.
\vspace{-1ex}
\\
\begin{minipage}{0.49\textwidth}
\begin{eqnarray*}
  \pguard{g} & \sdefs &\cpstep{g} \nondet \atomid 
\end{eqnarray*}
\end{minipage}
\begin{minipage}{0.49\textwidth}
\begin{eqnarray*}
  \Guar{g}    & \sdefs & \Om{\pguard{g}}
\end{eqnarray*}
\end{minipage}
\vspace*{1ex}
\\
A command $c$ with a guarantee of $g$ enforced on every program step
could possibly be expressed as $(\Guar{g}) \sqcup c$,
but that turns out to be too strong a requirement because it masks any aborting behaviour of $c$
because the guarantee never aborts, $(\Guar{g}) \sqcup \bot = (\Guar{g})$.
Instead, the weak conjunction operator  is used.

\emph{Weak conjunction} on commands, $\together$,
behaves like $\sqcup$ unless one of its operands aborts in which case we have
$c \together \bot = \bot$. The operator is associative, commutative and idempotent, and 
satisfies $c \together (\textstyle \Nondet D) = (\textstyle\Nondet_{d \in D} c \together d)$
for any non-empty set of commands $D$.
For any commands $c$ and $d$, steps $a$ and $b$, and tests $t$ and $t'$
weak conjunction satisfies the following axioms.
(Note the similarities between 
(\ref{A-together-interchange-sequential}) and (\ref{A-join-interchange-sequential}),
(\ref{A-atomic-together-nil}) and (\ref{A-atomic-join-nil})
and (\ref{A-atomic-together-inf}) and (\ref{A-atomic-iteration-infinite-infinite}).)
\\
\begin{minipage}{0.38\textwidth}
\begin{eqnarray}
  c \together \bot & = & \bot  \label{axiom-together-abort} \\
  a \together b & = & a \sqcup b \label{axiom-together-atomic} \\
  t \together t' & = & t \sqcup t' \label{axiom-together-test} 
\end{eqnarray}
\end{minipage}
\begin{minipage}{0.62\textwidth}
\begin{eqnarray}
  (a \SSeq c) \together (b \SSeq d) & = & (a \together b) \SSeq (c \together d) \label{A-together-interchange-sequential}\\
  (a \SSeq c) \together \Nil & = & \top \label{A-atomic-together-nil}\\
  \Inf{a} \together \Inf{b} &=& \Inf{(a \together b)}  \label{A-atomic-together-inf}
\end{eqnarray}
\end{minipage}\\
\\
Hence $a \together \cstepd = a \sqcup \cstepd = a$,
i.e.\ $\cstepd$ is the atomic step identity of weak conjunction.
More generally, $\Chaos \sdef \Om{\cstepd}$ is the identity of weak conjunction 
for any sequence of atomic steps.
The following lemma (and its proof) is similar to the corollary of Lemma~\ref{L-atomic-iteration-either}.
\begin{lemma}[atomic-iteration-conjunction]\label{L-atomic-iteration-conjunction}
\(
  \Om{a} \together \Om{b} = \Om{(a \together b)}
\)
\end{lemma}
A command $c$ with a guarantee $g$ is represented by $(\Guar{g}) \together c$.
In the theory of Jones, a guarantee on a process may be strengthened.
That is reflected by the fact that if $g_1 \subseteq g_2$, then $\cpstep{g_2} \refsto \cpstep{g_1}$
and hence $\pguard{g_2} \refsto \pguard{g_1}$.
A process that must satisfy both guarantee $g_1$ and guarantee $g_2$,
must satisfy $g_1 \cap g_2$ because
\begin{eqnarray*}
  &&\pguard{g_1} \together \pguard{g_2} \\
  & = & (\cpstep{g_1} \nondet \atomid) \together (\cpstep{g_1} \nondet \atomid) \\
  & = & (\cpstep{g_1} \together \cpstep{g_2}) \nondet (\cpstep{g_1} \together \atomid) \nondet (\atomid \together \cpstep{g_2}) \nondet (\atomid \together \atomid) \\
  & = & \cpstep{g_1 \cap g_2} \nondet \atomid \\
  & = & \pguard{(g_1 \cap g_2)}
\end{eqnarray*}

The weak conjunction of a possibly infinite iteration of
atomic steps distributes over the sequential composition of commands
$c$ and $d$.

\begin{lemma}[atomic-infinite-distribution]\label{L-atomic-infinite-distribution}
\(
\Om{a} \together (c\SSeq d) ~=~ (\Om{a} \together c) \SSeq (\Om{a} \together d)
\)
\end{lemma}

The proof 
uses the canonical representation of a command (Theorem~\ref{T-canonical-representation}) and
can be found in the appendix of \cite{TRatomicSteps}. 
As a consequence guarantees distribute over a sequence of commands.
\begin{eqnarray*}
  (\Guar{g}) \together (c \SSeq d) & = & ((\Guar{g}) \together c) \SSeq ((\Guar{g}) \together d)
\end{eqnarray*}

\subsection{The rely command}\label{S-rely}

A rely condition $r$ represents an assumption about environment steps.
If any environment step does not satisfy $r$, 
i.e.\ a step that refines $\cestep{\overline{r}}$,
the process may do anything,
which can be represented by it aborting.
Any other step is allowed.
The rely command is defined in terms of a single step assumption,
itself defined in terms of the abstract command $\Assume{}$ (\ref{def-assume}) as follows.
\begin{eqnarray*}
   \eassume{r} & \sdefs & \Assume{(\cnegate\cestep{\overline{r}})} 
               ~~~~~~~~ = \cnegate\cestep{\overline{r}} \nondet \cestep{\overline{r}} \SSeq \bot \\
   \Rely{r} & \sdefs & \Om{\eassume{r}}
\end{eqnarray*}
An environment assumption is placed on a command $c$ by placing the assumption on every step of $c$,
i.e.\ $(\Rely{r}) ~\together~ c$.
A command $c$ with rely $r$ and guarantee $g$ is expressed as
$(\Rely{r}) \together (\Guar{g}) \together c$,
for which every program step is required to satisfy $g$
unless an environment step does not satisfy $r$, in which case it aborts.
Here using weak conjunction ($\together$) rather than the lattice join ($\sqcup$) is essential
to prevent the guarantee masking the 
possible aborting behaviour of the rely. 
Because $\Assume{a} \together \Assume{b} = \Assume{(a \sqcup b)}$,
combining environment assumptions gives
\begin{eqnarray*}
  \eassume{r_1} \together \eassume{r_2}
  & = & \Assume{(\cnegate{\cestep{\overline{r_1}}} \sqcup \cnegate{\cestep{\overline{r_2}}})} 
  ~ = ~ \Assume{(\cnegate{\cestep{\overline{r_1 \cap r_2}}})} \\
  & = & \eassume{(r_1 \cap r_2)}~.
\end{eqnarray*}

From Lemma~\ref{L-atomic-infinite-distribution} and Theorem~\ref{T-canonical-representation}, a rely can be distributed over a sequential composition (the proof is included in the appendix of \cite{TRatomicSteps}).
\begin{eqnarray*}
  (\Rely{r}) \together (c \SSeq d) & = & (\Rely{r} \together c) \SSeq (\Rely{r} \together d)
\end{eqnarray*}

\subsection{Rely/Guarantee Logic}
\label{S-rg-logic}

\sloppy
Rely/guarantee reasoning is traditionally formulated in terms of a quintuple 
$\quintprgqc$, which extends Hoare logic with the rely $r$ and guarantee $g$ to handle
concurrency.  The quintuple states that every step of $c$ satifies $g$ and that it terminates and
establishes the postcondition $q$, provided it is executed from an initial state satisfying $p$ and
interference from the environment is bounded by $r$.  This quintuple is interpreted in our logic as
the following refinement.%
\footnote{
We use the syntax of Morgan's specification command $[q]$ \cite{TSS} whose definition 
is omitted for space reasons. It represents any sequence of
atomic steps that establishes $q$ between its initial and final states.  
See \cite{DaSMfaWSLwC-TR} for details.
}
\begin{equation*}
	\Pre{p} \SSeq ((\Rely{r}) \together (\Guar{g}) \together [q]) \refsto c
\end{equation*}
This demonstrates the application of the algebra to reasoning about shared
data.
As well as being able to express any law presented in terms of quintuples,
we are able to reason about the component commands separately,
e.g., strengthening a guarantee $g$ does not involve $p$, $r$ and $q$.

%% file: processAlgebrasUsingAtomicSteps.tex
\section{Abstract communication in process algebras}\label{S-process-algebras}

In the process algebra domain, processes communicate via a set of synchronisation events, in contrast to 
processes in a shared memory concurrency model which
interleave operations on state.
We may build a core process algebra from the basic operators, with the addition of
a set of atomic program steps $\synchpstepa$
that model a process engaging in the corresponding abstract event $a \in \Event$, 
where $\Event$ includes at least the silent event $\silent$.
The basic properties of this language are those of the underlying algebra 
but we do not assume conjunctivity of sequential composition (\ref{L-seq-distr-left})
in order to be consistent with CCS.

Similarly to notation introduced in Section~\ref{S-atomic-iteration} we define
\begin{equation}
	\ateva \sdef 
		\Om{\synchestepE}
		\synchpstepa
		\Om{\synchestepE}
\end{equation}
This models process engaging in event $a$ (note that we drop the `$\pi$' tag from the $\ateva$ notation) 
preceded and succeeded by steps of the environment, similar to 
\emph{asynchronising} in Synchronous CCS \cite{Milner83} (discussed in \cite{CaC}).
This is the building block of event based languages:
we interpret both prefixing in CCS ($a.p$) and CSP ($a \fun p$) as $(\ateva \SSeq p)$.
We extend the core algebra to give two types of abstract interprocess communication: 
CCS-style binary synchronisation
(achieved by restricting the program)
and 
CSP-style multi-way synchronisation 
(achieved in-part by restricting the environment).

\subsection{Communication in CCS}
\label{S-ccs-encoding}

The main point of difference with the rely-guarantee algebra is that program steps representing
events can combine into a single program step
(communication).  
Interactions with $\sestepE$ remain the same as in the abstract algebra.
In CCS each non-silent event $a$ has a complementary event $\ecompla$.
A program step $\spstepa$ and its corresponding complementary program step $\spstepac$
may synchronise to become a silent step, $\spstepa \pl \spstepac = \ssilent$,
and hence using an instantiation of \reflemma{atomic-interleaving},
\begin{eqnarray}
  \ateva \pl \atevac ~ = ~ \atevs \nondet \ateva \SSeq \atevac \nondet \atevac \SSeq \ateva~.
  \label{eqn-ccs-atev-sync}
\end{eqnarray}
As such, events may synchronise \emph{or} interleave.
In CCS the restriction operator 
$\Res{A}{p}$,
where $A$ is a set of $\Event$s,
may be employed to exclude the final two interleaving options and hence
force processes to synchronise and generate a silent step.
It may be defined straightforwardly using join ($\sqcup$) to forbid events
in $A$, where we use the abbreviation
$
	\spstepA \sdef \Nondet_{a \in A} \synchpstepa
$ and note that $\cnegate \spstepA = \spstep{\overline{A}} \nondet \sestepE$.
\begin{equation}
\label{eqn-ccs-res}
	\Res{A}{p} \sdef p \sqcup \Om{(\cnegate \spstepA)}
\end{equation}
Hence, by \refeqn{ccs-atev-sync} and \refeqn{ccs-res},
$
  \Res{\{a, \ecompla\}}{(\ateva \pl \atevac)} 
  ~ = ~ 
  \atevs
$.

\subsection{Communication in CSP}
\label{S-csp-encoding}

To achieve CSP-style multi-way communication, a process $p$ prevents its environment from communicating via an
event in $p$'s alphabet until $p$ is ready.  We introduce a step $\sestepa$, where $\sestepE \refsto
\sestepa$ for all $a \in \Event$.
Its interactions through the parallel operator are defined (in a different way to CCS) below; all
other combinations of atomic steps result in $\top$.
\vspace*{-1ex}
\begin{equation*}
	\spstepa \pl \spstepa ~=~ \spstepa 
	\quad \mbox{for $a \neq \silent$}
	\quad
	\quad
	\spstepa \pl \sestepa ~=~ \spstepa 
	\quad
	\quad
	\sestepa \pl \sestepa ~=~ \sestepa 
\end{equation*}

Fundamental to CSP is the notion of a process's \emph{alphabet}, the set of events via which it may communicate and
in particular upon which the environment may not independently synchronise.
Here we explicitly associate an alphabet $A \subseteq Event$ with process $p$ by the syntax $\AProcAp$, defined by, 
\vspace*{-1ex}
\begin{equation}
\label{eqn-csp-alphabet}
	\AProcAp \sdef p \sqcup \Om{(\cnegate{\sestep{A}})}
\end{equation}
\vspace*{-3ex}
\\
where
analogously to program steps we define $\sestepA \sdef \Nondet_{a \in A} \sestepa$.
Note the similarity to CCS's restriction operator \refeqn{ccs-res}
but here it is the environment that is restricted, rather than the program.

In an early formulation by Hoare \cite{Hoare85} every process $p$ implicitly has an alphabet $A$
associated with it ($A$ is sometimes syntactically deduced from $p$).  
In formulations such as Roscoe's \cite{Roscoe98}
the alphabets are not associated with processes but are instead made explicit on the parallel operator.
We may define alphabetised parallel straightforwardly as
$
	p_1 \plA p_2
	\sdef
	(\AProcA{p_1}) \pl (\AProcA{p_2})
$.
Each side of the parallel composition prevents the other from taking a unilateral program step on events in $A$
by restricting its environment.  
Some of the basic communication properties from CSP follow from the above definitions and the atomic algebra, 
for instance, recalling that CSP's prefixing operator $a \fun p \sdef \ateva \SSeq p$,
for any $a \in A$,
$
	(a \fun p_1 \plA a \fun p_2) = a \fun (p_1 \plA p_2)
$.

The \emph{hiding} operator of CSP, $\Hide{A}{p}$, affects program steps, renaming events in $A$
to silent events.  
Hiding distributes over sequential and choice (but not parallel); its
relationship with atomic steps is 
$
	\Hide{A}{b} = \left \{ \begin{array}{ll}
						\ssilent & \mbox{if $b$ is of the form $\spstepa$ and $a \in A$}\\
						b & \mbox{otherwise}
						\end{array}
						\right.
$.

\subsection{Communication in SCCS}

Synchronous CCS (SCCS) \cite{Milner83,CaC} is a process algebra designed to be as minimal as possible in terms of operators.
It includes event prefix, disjunction (nondeterministic choice), composition (corresponding to our parallel), and restriction similar to that of CCS
\refeqn{ccs-res}.
SCCS events may be structured from  a finite set of ``particles'', 
forming a commutative group $(Event, \One, \mult, \Inverse{})$.  Every event is the product of particles: for instance, the
step $a$ is an event $(a^1 \mult b^0 \mult c^0 \mult \ldots)$.  The silent (or waiting) event $\One$ is event identity,
and fulfils a similar role to that of $\atomid$ in our algebra.  The complement of event $a$ is simply $\Inverse{a}$
and hence the product of an event and its complement, $a^1 \mult \Inverse{a}$, naturally equals $\One$.

The key aspect of SCCS is its simple definition of parallel composition in terms of product:
for atomic steps $a$ and $b$, $a \pl b = a \mult b$.  An 
event process
$\atev{a}$ is defined as $\Om{\One} \SSeq a \SSeq \Om{\One}$, which has the effect of \emph{asynchronising} the event, preserving \reflemma{atomic-interleaving}.
Milner shows that CCS can be encoded in SCCS through the
addition of asynchronising actions defined through the operational semantics; in an algebraic
setting the $\One$s are made explicit in the processes.  
Note that in this model there is no distinction between silent steps 
and environment steps: in SCCS both are $\One$, whereas in CCS the former is $\ssilent$.

%% file: conclusion.tex
\section{Related Work}\label{S-related-work}

Our Concurrent Refinement Algebra (CRA) (Section~\ref{S-general-algebra}) 
compares to Concurrent Kleene Algebra (CKA)
\cite{DBLP:journals/jlp/HoareMSW11} in that both extend a sequential algebra to allow
for reasoning about parallel composition. 
Synchronous Kleene Algebra (SKA) \cite{Pris10} is also based on Kleene Algebra 
but, unlike CKA, it adds tests and a synchronous parallel operator based on that of
Milner's SCCS \cite{Milner83}.
Both CKA and SKA are based on Kleene algebra and 
hence only support finite iteration and partial correctness. 
In comparison, our CRA
supports general fixed points 
and hence recursion and both finite and infinite iteration.
The richer structure of DRA contains a sub-lattice of commands below $\Chaos$
(see Fig.\ \ref{lattices}) that includes 
assertions (and hence preconditions in the relational interpretation)
and
assumptions (and hence rely commands),
and allows the weak conjunction operator, $\together$, to be distinguished from strong conjunction, $\sqcup$.
All these constructs are needed to faithfully represent rely/guarantee theory.

CKA is also applied to rely/guarantee rules \cite{DBLP:journals/jlp/HoareMSW11}
but they define a Jones-style 5-tuple  (as in Section~\ref{S-rg-logic})
in terms of two separate refinement conditions,
whereas in our approach the existing (single) refinement relation can be used directly.
In Jones' theory, a guarantee has to be satisfied only from initial states satisfying the precondition of the program,
and further, if its rely condition is broken by the environment, the program can abort.
However, in the CKA framework, the guarantee has to always be maintained by the program, irrespective
of what the initial state is and how the environment is behaving;
that over restricts the set of possible implementations.
Our theory faithfully reflects Jones' approach.

Our algebra of atomic steps makes use of a synchronous parallel operator similar to that 
in SCCS \cite{CaC} and 
in SKA \cite{Pris10}
but 
it differs in two ways: 
\begin{itemize}
\item
instead of atomic actions being separate from commands (as in SCCS and SKA),
they are treated as a sub-algebra within CRA
and
\item
while both 
SCCS and SKA
explicitly define composition of atomic steps (their $\times$ operator),
our parallel operator is used directly on atomic steps (because they are commands)
and its definition is left open.
\end{itemize}

\section{Conclusion}\label{S-conclusion}

This paper presents an abstract algebra of atomic steps for concurrent programs.
It is a Boolean algebra that is embedded as a sub-lattice into our Concurrent Refinement Algebra 
in a similar way as tests are embedded in Kleene algebras.
As for tests, a range of useful laws 
can be derived for atomic steps within this abstract algebra (e.g., on iteration and distributivity),
despite the fact that the interpretation of the parallel composition of two atomic steps is left open.

This construction simplifies many essential laws and
their proofs, as most supporting lemmas almost come for free on this abstract level. 
Accordingly, the mechanisation of the theory within the theorem prover Isabelle is lean 
and achieved a high degree of automation.
As the Concurrent Refinement Algebra was conceived to support reasoning with
relies and guarantees this simplification is of particular benefit in our laws
for rely and guarantee commands. 

A further gain of the generic shape of the abstract algebra lies in
its potential for reuse. We have demonstrated this by instantiating our abstract algebra
with two quite different styles of communication, a synchronous model (as in SKA \cite{Pris10} and SCCS) 
versus an interleaving model (as in CCS and CSP). 
For both styles the abstract algebra of atomic steps proves to be 
suitable.

The concept of sub-algebras in our Concurrent Refinement Algebra
is also applicable to assertions and assumptions. Assertions form a Boolean algebra 
with $\Nil$ as top element and $\bot$ as bottom element whereas 
step assumptions form a Boolean algebra 
with top element $\cstepd$ and bottom $\cstepd \SSeq \bot$.
Both inherit the laws on Boolean algebras similarly to 
tests and atomic steps.
Future work will investigate these structures and will extend our theories accordingly. 

\sloppy
The relationship between CCS and CSP has been explored in several papers
\cite{Brookes83,Glabbeek97} including augmenting the operational rules of
CSP so that the failures-divergences model (FDR) is respected in CCS
\cite{CSP-retractOf-CCS}.  
Future work is to apply a
more algebraic approach to the relationships between well known process
algebras (especially ACP \cite{BergstraKlop84}).

%% file: proofs.tex
\newenvironment{lemmaproof}[1]{\par\noindent\textbf{Lemma \ref{L-#1} (#1)}\itshape}{}
\newenvironment{theoremproof}[1]{\par\noindent\textbf{Theorem \ref{T-#1} (#1)}\itshape}{}
\def\arraystretch{1.5}

\newenvironment{new-proof}{\paragraph{{\bf Proof.}}}{\hfill$\Box$}

\section{Proofs for inspection}

For all lemmas we assume $a$ and $b$ to be atomic steps and $c$ and $d$ arbitrary commands.
Furthermore, all lemmas except Lemmas~\ref{L-atomic-identity-iteration} and ~\ref{L-atomic-infinite-distribution} assume that 
the sequential operator is conjunctive (\ref{L-seq-distr-left}) as this property is used within
the proofs. \vspace*{2ex}

\begin{lemmaproof}{atomic-iteration-finite}%\label{L-atomic-iteration-finite}
\begin{eqnarray*}
 \Fin{a} \SSeq c \parallel \Fin{b} \SSeq d 
  & = & \Fin{(a \parallel b)} \SSeq ((c \parallel d) \nondet (c \parallel b \SSeq \Fin{b} \SSeq d) \nondet (a \SSeq \Fin{a} \SSeq c \parallel d))
\end{eqnarray*}
\end{lemmaproof}

\begin{proof}
The proof relies on (\ref{L-finite-iteration}), i.e., $\Fin{a} = \Nondet_{i \in \nat} a^i$. 
The notation $\Nondet_{i,j \in \nat}^{i < j} c_{i,j}$ stands for the choice of $c_{i,j}$
 over all natural numbers $i$ and $j$, such that $i < j$.
\begin{displaymath}
 \begin{array}{cl}
\multicolumn{2}{l}{\Fin{a} \SSeq c \parallel \Fin{b} \SSeq d} \\[1ex]
    = & (\Nondet_{i \in \nat} a^i \SSeq c) \parallel (\Nondet_{j \in \nat} b^j \SSeq d)
      ~=~ \Nondet_{i \in \nat, j \in \nat} (a^i \SSeq c \parallel b^j \SSeq d)\\[1ex]
   = & \Nondet_{i \in \nat} (a^i \SSeq c \parallel b^i \SSeq d) 
         \nondet \Nondet_{i,j \in \nat}^{i < j} (a^i \SSeq c \parallel b^i \SSeq b^{j-i} \SSeq d)
         \nondet \Nondet_{i,j \in \nat}^{i > j} (a^j \SSeq a^{i-j} \SSeq c \parallel b^{j} \SSeq d)\\[1ex]
   = & \Nondet_{i \in \nat} (a \parallel b)^i \SSeq (c \parallel d) 
         \nondet \Nondet_{i, k \in \nat}^{k > 0}(a \parallel b)^i \SSeq (c \parallel b^k \SSeq d) 
         \nondet \Nondet_{j, k \in \nat}^{k > 0}(a \parallel b)^j \SSeq (a^k \SSeq c \parallel d)\\[1ex]
   = & \Fin{(a \parallel b)}\SSeq (c \parallel d) 
         \nondet (\Nondet_{i \in \nat} (a \parallel b)^i)  \SSeq \Nondet_{k \in \nat}^{k > 0}(c \parallel b^k \SSeq d)\\[1ex]
     &~    \nondet (\Nondet_{j \in \nat} (a \parallel b)^j) \SSeq \Nondet_{k \in \nat}^{k > 0}(a^k \SSeq c \parallel d) \\[1ex]
   = &  \Fin{(a \parallel b)} \SSeq ((c \parallel d) \nondet (c \parallel \Nondet_{k \in \nat}^{k > 0}b^k \SSeq d))
         \nondet (\Nondet_{k \in \nat}^{k > 0}a^k \SSeq c \parallel d)\\[1ex]
   = & \Fin{(a \parallel b)} \SSeq ((c \parallel d) \nondet (c \parallel b \SSeq \Fin{b} \SSeq d) \nondet (a \SSeq \Fin{a} \SSeq c \parallel d))
 \end{array}
\end{displaymath}
\end{proof}

\begin{lemmaproof}{atomic-iteration-finite-infinite}%\label{L-atomic-iteration-finite-infinite}
\[
 ~~\Fin{a} \SSeq c \parallel \Inf{b} = \Fin{(a \parallel b)} \SSeq (c \parallel \Inf{b})
\]
\end{lemmaproof}

\begin{proof}
Note that, by unfolding law (\ref{L-infinite-unfold-power}), $\Inf{b} = b^i \Inf{b}$ for any $i \in \nat$.
The proof uses also Lemma~\ref{L-atomic-iteration-power}.
\begin{displaymath}
 \begin{array}{rclclcl}
\Fin{a} \SSeq c \parallel \Inf{b} 
  & = & (\Nondet_{i \in \nat} a^i \SSeq c) \parallel \Inf{b}\\
  & = & \Nondet_{i \in \nat}(a^i \SSeq c \parallel \Inf{b}) \\
  & = & \Nondet_{i \in \nat}(a^i \SSeq c \parallel b^i \SSeq \Inf{b}) \\
  & = &  \Nondet_{i \in \nat}(a \parallel b)^i \SSeq (c \parallel \Inf{b})\\
  & = & \Fin{(a \parallel b)} \SSeq (c \parallel \Inf{b})  
 \end{array}
\end{displaymath}
\end{proof}

\begin{lemmaproof}{atomic-iteration-either}%\label{L-atomic-iteration-either}
\begin{eqnarray*}
 \Om{a} \SSeq c \parallel \Om{b} \SSeq d & = &
    \Om{(a \parallel b)} \SSeq ((c \parallel d) \nondet (c \parallel b \SSeq \Om{b} \SSeq d) \nondet (a \SSeq \Om{a} \SSeq c \parallel d))
\end{eqnarray*}
\end{lemmaproof}

\begin{proof} Note that, by (\ref{L-isolation}) and (\ref{L-infinite-annihilates}), 
$\Om{a} = \Fin{a} \nondet \Inf{a}$ and $\Inf{a} c = \Inf{a}$, 
and by (\ref{L-iteration2}), $a \SSeq \Fin{a} \SSeq c \nondet \Inf{a} = a \SSeq \Om{a} \SSeq c$.
The proof uses also (\ref{L-iteration1}), 
and Lemmas \ref{L-atomic-iteration-finite} and \ref{L-atomic-iteration-finite-infinite}, 
and (\ref{A-atomic-iteration-infinite-infinite}), i.e.\ $\Inf{a} \parallel \Inf{b} = \Inf{(a \parallel b)}$.

\begin{displaymath}
 \begin{array}{rcl}
 \multicolumn{3}{l}{
  \Om{a} \SSeq c \parallel \Om{b} \SSeq d } \\
  & = & (\Fin{a} \nondet \Inf{a}) \SSeq c \parallel (\Fin{b} \nondet \Inf{b}) \SSeq d \\
  & = & (\Fin{a} \SSeq c \parallel \Fin{b} \SSeq d) \nondet (\Fin{a} \SSeq c \parallel \Inf{b}) \nondet 
           (\Inf{a} \parallel \Fin{b} \SSeq d) \nondet (\Inf{a} \parallel \Inf{b}) \\
  & = & \Fin{(a \parallel b)} \SSeq ((c \parallel d) \nondet (c \parallel b \SSeq \Fin{b} \SSeq d) \nondet (a \SSeq \Fin{a} \SSeq c \parallel d)) \nondet {} \\
  &    & \Fin{(a \parallel b)} \SSeq (c \parallel \Inf{b}) \nondet \Fin{(a \parallel b)} \SSeq (\Inf{a} \parallel d) \nondet \Inf{(a \parallel b)} \\
  & = & \Fin{(a \parallel b)} \SSeq ((c \parallel d) \nondet ((c \parallel b \SSeq \Fin{b} \SSeq d) \nondet (c \parallel \Inf{b})) \nondet {} \\
  &    & ((a \Fin{a} \SSeq c \parallel d) \nondet (\Inf{a} \parallel d))) \nondet \Inf{(a \parallel b)} \\
  & = & \Fin{(a \parallel b)} \SSeq ((c \parallel d) \nondet (c \parallel (b \SSeq \Fin{b} \SSeq d \nondet \Inf{b})) \nondet {}
            ((a \SSeq \Fin{a} \SSeq c  \nondet \Inf{a}) \parallel d))) \nondet \Inf{(a \parallel b)} \\
  & = & \Fin{(a \parallel b)} \SSeq ((c \parallel d) \nondet (c \parallel b \SSeq \Om{b} \SSeq d) \nondet
            (a \SSeq \Om{a} \SSeq c  \parallel d)) \nondet \Inf{(a \parallel b)} \\
  & = & \Om{(a \parallel b)} \SSeq ((c \parallel d) \nondet (c \parallel b \SSeq \Om{b} \SSeq d) \nondet (a \SSeq \Om{a} \SSeq c  \parallel d))
 \end{array}
\end{displaymath}
\end{proof}

\begin{lemmaproof}{atomic-identity-iteration}%\label{L-atomic-identity-iteration}
\(
  ~~~\Om{\atomid} \parallel c = c
\)
\end{lemmaproof}
\vspace*{2ex}

\begin{proof}
Theorem~\ref{T-canonical-representation} states that $c$ can be represented as
$t \nondet t' \SSeq \bot \nondet \Nondet_{i \in I} a_i \SSeq c_i$.
The proof is via structural induction and hence we assume $\Om{\atomid} \parallel c_i = c_i$, for all $i \in I$.
\begin{eqnarray*}
  \Om{\atomid} \parallel c 
  & = & \Om{\atomid} \parallel (t \nondet t' \SSeq \bot \nondet \Nondet_{i \in I} a_i \SSeq c_i) \\
  & = & (\Om{\atomid} \parallel t) \nondet (\Om{\atomid} \parallel t' \SSeq \bot) \nondet (\Om{\atomid} \parallel \Nondet_{i \in I} a_i \SSeq c_i) \\
  & = & t \nondet t' \SSeq \bot \nondet \Nondet_{i \in I} (\Om{\atomid} \parallel a_i \SSeq c_i) \\
  & = & t \nondet t' \SSeq \bot \nondet \Nondet_{i \in I} (\atomid \SSeq \Om{\atomid} \parallel a_i \SSeq c_i) \\
  & = & t \nondet t' \SSeq \bot \nondet \Nondet_{i \in I} a_i \SSeq (\Om{\atomid} \parallel \SSeq c_i) \\
  & = & t \nondet t' \SSeq \bot \nondet \Nondet_{i \in I} a_i \SSeq c_i ~~~\mbox{by inductive hypothesis} \\
  & = & c
\end{eqnarray*}
Note that 
$\Om{\atomid} \parallel t = t \SSeq (\Om{\atomid} \parallel \Nil) = t \SSeq \Nil = t$
and
$\Om{\atomid} \parallel t' \SSeq \bot = t' \SSeq (\Om{\atomid} \parallel \bot) = t' \SSeq \bot$
because $\Om{\atomid} \parallel \bot \refsto \atomid \parallel \bot = \bot$ 
and hence $\Om{\atomid} \parallel \bot = \bot$.

\end{proof}

\begin{lemmaproof}{atomic-interleaving}%\label{L-atomic-interleaving}
\[
  ~~\atomic{a} \parallel \atomic{b} ~=~ \atomic{a \parallel b} \nondet \atomic{a} \SSeq \atomic{b} \nondet \atomic{b} \SSeq \atomic{a}
\]
\end{lemmaproof}

\begin{proof}
The proof uses Lemmas \ref{L-atomic-iteration-either} and \ref{L-atomic-identity-iteration} 
and $\Om{\atomid} \SSeq \Om{\atomid} = \Om{\atomid}$.
\begin{displaymath}
 \begin{array}{rcl}
  \multicolumn{3}{l}{\atomic{a} \parallel \atomic{b}} \\
  & = & (\Om{\atomid} \SSeq a \SSeq \Om{\atomid}) \parallel (\Om{\atomid} \SSeq b \SSeq \Om{\atomid}) \\
  & = & (\Om{\atomid} \SSeq (a \SSeq \Om{\atomid} \parallel b \SSeq \Om{\atomid}) \nondet 
            (\Om{\atomid} \SSeq (a \SSeq \Om{\atomid} \parallel \atomid \SSeq \Om{\atomid} \SSeq b \SSeq \Om{\atomid})) \nondet 
            (\Om{\atomid} \SSeq (\atomid \SSeq \Om{\atomid} \SSeq a \SSeq \Om{\atomid} \parallel b \SSeq \Om{\atomid})) \\
  & = & (\Om{\atomid} \SSeq (a \parallel b) \SSeq \Om{\atomid}) \nondet 
            (\Om{\atomid} \SSeq a \SSeq (\Om{\atomid} \parallel \Om{\atomid} \SSeq b \SSeq \Om{\atomid})) \nondet 
            (\Om{\atomid} \SSeq b \SSeq (\Om{\atomid} \SSeq a \SSeq \Om{\atomid} \parallel \Om{\atomid})) \\
  & = & (\Om{\atomid} \SSeq (a \parallel b) \SSeq \Om{\atomid}) \nondet
            (\Om{\atomid} \SSeq a \SSeq \Om{\atomid} b \SSeq \Om{\atomid}) \nondet
            (\Om{\atomid} \SSeq b \SSeq \Om{\atomid}a \SSeq \Om{\atomid}) \\
  & = & \atomic{a \parallel b} \nondet \atomic{a} \SSeq \atomic{b} \nondet \atomic{b} \SSeq \atomic{a}
 \end{array}
\end{displaymath}
\end{proof}

\begin{lemmaproof}{atomic-infinite-distribution}%\label{L-atomic-infinite-distribution}
\[
  ~~\Om{a} \together (c\SSeq d) ~=~ (\Om{a} \together c) \SSeq (\Om{a} \together d)
\]
\end{lemmaproof} 

\begin{proof}
Using Theorem~\ref{T-canonical-representation} 
we may assume \(c = t \nondet t' \SSeq \bot \nondet \Nondet_{i \in I} b_i \SSeq c_i\).
The proof follows by induction on command $c$, 
i.e. assume the lemma holds for $c_i$. 
\begin{displaymath}
 \begin{array}{rcl}
   \multicolumn{3}{l}{\Om{a} \together (c\SSeq d)} \\
   & = & \Om{a} \together ( (t \nondet t'\bot \nondet \Nondet_{i \in I} b_i\SSeq c_i)\SSeq d)\\
   & = & \Om{a} \together (t \SSeq d \nondet t'\bot \nondet \Nondet_{i \in I} b_i\SSeq c_i \SSeq d)\\
   & = & (\Om{a} \together t \SSeq d) \nondet (\Om{a} \together t' \SSeq \bot) \nondet
          (\Om{a} \together  \Nondet_{i \in I} b_i\SSeq c_i \SSeq d)\\
   & = & t \SSeq (\Om{a} \together d) \nondet t' \SSeq \bot \nondet
          (\Om{a} \together  \Nondet_{i \in I} b_i\SSeq c_i \SSeq d)\\
   & = & t \SSeq (\Om{a} \together d) \nondet t' \SSeq \bot \SSeq (\Om{a} \together d) \nondet
          \Nondet_{i \in I} (a \together b_i) \SSeq (\Om{a} \together  c_i \SSeq d)\\
   & = & (t \nondet t'\SSeq \bot ) \SSeq (\Om{a} \together d) \nondet
          \Nondet_{i \in I} (a \together b_i) \SSeq (\Om{a} \together  c_i )\SSeq (\Om{a} \together  d) \\
   & = & (t \nondet t'\SSeq \bot ) \SSeq (\Om{a} \together d) \nondet
          \Nondet_{i \in I} (a \SSeq \Om{a} \together b_i \SSeq c_i) \SSeq (\Om{a} \together  d) \\
   & = & (t \nondet t'\SSeq \bot \nondet
          \Nondet_{i \in I} (\Om{a} \together b_i \SSeq c_i))\SSeq (\Om{a} \together  d) \\
   & = & ((\Om{a} \together t) \nondet (\Om{a} \together t'\SSeq \bot) \nondet
          (\Om{a} \together \Nondet_{i \in I} b_i \SSeq c_i))\SSeq (\Om{a} \together  d) \\
   & = & (\Om{a} \together (t \nondet t'\SSeq \bot \nondet \Nondet_{i \in I} b_i \SSeq c_i))
         \SSeq (\Om{a} \together  d) \\
   & = & (\Om{a} \together c)
         \SSeq (\Om{a} \together  d) \\
 \end{array}
\end{displaymath}
\end{proof}

\newpage

\begin{lemma}\label{atom-inf-conj-test}
For any test $t$, ~~
\( \Om{a} \together t = t \)
\end{lemma}
\begin{proof}
\begin{displaymath}
 \begin{array}{l}
 \Om{a} \together t
\Equals (\Nil \nondet a\Om{a}) \together t
\Equals (\Nil \together t) \nondet ( a\Om{a} \together t)
\Equals t \nondet t(a\Om{a} \together \Nil) ~=~ t \nondet t\top ~=~ t
\end{array}
\end{displaymath} 
\end{proof}

\begin{lemma}\label{atom-inf-conj-test-bot}
For any test $t$, ~~
\(  \Om{a} \together t\bot = t\bot \)
\end{lemma} 
\begin{proof}
\begin{displaymath}
 \begin{array}{l}
 \Om{a} \together t\bot
\Equals (\Nil \nondet a\Om{a}) \together t\bot
\Equals (\Nil \together t\bot) \nondet ( a\Om{a} \together t\bot)
\Equals t(\Nil \together \bot) \nondet t(a\Om{a} \together \bot) ~=~ t\bot \nondet t\bot ~=~ t\bot
\end{array}
\end{displaymath} 
\end{proof}

\begin{lemma}[rely-distribution]
\[
  ~~(\Rely{r}) \together (c\SSeq d) ~=~ ((\Rely{r}) \together c) \SSeq ((\Rely{r}) \together d)
\]
\end{lemma} 

\begin{proof}
We base the proof on the more general concept of \emph{assumptions}. Using the 
definition of environment assumptions, $\eassume{r} = \Assume{(\cnegate\cestep{\overline{r}})} = \cnegate\cestep{\overline{r}} \nondet \cestep{\overline{r}}\SSeq \bot$ 
and relies, $\Rely{r} \sdef \Om{\eassume{r}}$, we can deduce  
\[\exists a \dot \Rely{r} = \Om{(\Assume{a})} = \Om{a} \nondet \Om{a}\SSeq \anegate{a} \SSeq \bot\]
Furthermore, using Theorem~\ref{T-canonical-representation} 
we may assume \(c = t \nondet t' \SSeq \bot \nondet \Nondet_{i \in I} b_i \SSeq c_i\).
\begin{displaymath}
 \begin{array}{l}
   \Om{(\Assume{a})} \together (c\SSeq d)
\Equals
  (\Om{a} \nondet \Om{a} \SSeq \anegate{a}\SSeq \bot) \together (c \SSeq d)
\Equals
  \Om{a}(\Nil \nondet \anegate{a}\bot) \together (c \SSeq d)\\
\Equals*[~~assume $x \sdef \Nil \nondet \anegate{a}\bot$ and the canonical form of $c$]
 \Om{a}x \together (t \nondet t' \SSeq \bot \nondet \Nondet_{i \in I} b_i \SSeq c_i)d
\Equals*[~~ distribute $\nondet$ over $\together$ and  left-distributivity of sequential]
 \Om{a}x \together t d ~\nondet~
 \Om{a}x \together t' \bot  ~\nondet~
 (\Nondet_{i \in I}  \Om{a}x \together b_i \SSeq c_i d)
\end{array}
\end{displaymath}
\begin{displaymath}
 \begin{array}{l}
\Equals*[by unfolding of $\Om{a}$ using (\ref{L-omega-unfold})]
 t(\Om{a}x \together d) ~\nondet~
 t' \bot  ~\nondet~ (\Nondet_{i \in I}  
  \fbox{$(a\Om{a}x \nondet x) \together b_i \SSeq c_i d)$}\vspace*{1ex}\\
\hspace*{2cm}
\begin{array}[t]{l}
  (a\Om{a}x \nondet x) \together b_i \SSeq c_i d)
 \Equals*[by definition of $x$]
    (a \together b_i)(\Om{a}x \together c_i d) \nondet (\anegate{a}\together b_i)(\bot \together c_i d)
  \Equals*[by induction assumption: $\Om{a}x \together c_i d = (\Om{a}x \together c_i)(\Om{a}x \together d)$]
   (a \together b_i)(\Om{a}x \together c_i)(\Om{a}x \together d) \nondet
      (\anegate{a}\together b_i)\bot
  \Equals*[with (\ref{axiom-together-abort}) and $\bot$ left annihilator ]%$\bot = (\bot \together c_i) (\Om{a}x \together d)$]
    (a \together b_i)(\Om{a}x \together c_i)(\Om{a}x \together d) \nondet
     (\anegate{a}\bot\together b_i c_i) (\Om{a}x \together d)
  \Equals*[with (\ref{A-atomic-together-nil})] %$\Nil \together b_i c = \top$]
   (a\Om{a}x \together b_i c_i)(\Om{a}x \together d) \nondet
    ((\anegate{a}\bot \nondet \Nil)\together b_i c_i) (\Om{a}x \together d)
  \Equals*[with (\ref{L-seq-distr-right})] %distribute sequential over $\nondet$]
   ((a\Om{a}x \together b_i c_i) \nondet (\anegate{a}\bot \nondet \Nil)\together b_i c_i))
    (\Om{a}x \together d)
  \Equals*[distribute $\together$ over $\nondet$, with definition of $x$ and (\ref{L-omega-unfold})]
    (\Om{a}x \together b_i c_i) (\Om{a}x \together d)
  \end{array}\vspace*{1ex}
\Equals
  t(\Om{a}x \together d) ~\nondet~
  t' \bot (\Om{a}x \together d) ~\nondet~
  (\Nondet_{i \in I}  \fbox{$(\Om{a}x \together b_i c_i) (\Om{a}x \together d)$})
\Equals*[distribute of sequential over $\nondet$]
  (t \nondet t'\bot \nondet (\Om{a}x \together \Nondet_{i \in I} b_i c_i)) (\Om{a}x \together d)
\Equals*[by Lemmas \ref{atom-inf-conj-test} and \ref{atom-inf-conj-test-bot}]
   ( (\Om{a}x \together t) \nondet  (\Om{a}x \together t'\bot) \nondet 
     (\Om{a}x \together \Nondet_{i \in I} b_i c_i)) (\Om{a}x \together d)
\Equals*[with Theorem~\ref{T-canonical-representation}]
 (\Om{a}x \together c) (\Om{a}x \together d)
\end{array}
\end{displaymath}
\end{proof}